\documentclass[10pt,journal,doublecolumn]{IEEEtran}%peerreview, version 5
\IEEEoverridecommandlockouts
\usepackage{algpseudocode}
\usepackage{algorithm}
\usepackage{bm}
\usepackage{amsfonts}
\usepackage{amssymb}
\usepackage{amscd}
\usepackage{graphics}
\usepackage[dvips]{graphicx}
\usepackage{epsfig}
\usepackage{subfigure}
\usepackage[cmex10]{amsmath}
\usepackage{array}
\usepackage{eqparbox}
\usepackage{flushend}
\usepackage{hyperref}
\usepackage{fancyhdr}
\usepackage{titling}
\usepackage{enumerate}
\usepackage{stackrel}
\usepackage{pdflscape}
\usepackage{color}
\usepackage{longtable}
\usepackage{siunitx}
\usepackage{booktabs}

\DeclareMathAlphabet{\mathbsf}{OT1}{cmss}{bx}{n}% bold sans serif
\DeclareMathAlphabet{\mathssf}{OT1}{cmss}{m}{sl}% slanted sans serif
\DeclareMathAlphabet{\mathcsf}{OT1}{cmss}{sbc}{n}% condensed sans serif

\newcommand{\ie}{{\em i.e.}}
\newcommand{\etc}{{\em etc}}

\newcommand{\secref}[1]{Section~\ref{#1}}
\newcommand{\figref}[1]{Fig.~\ref{#1}}
\newcommand{\tabref}[1]{Table~\ref{#1}}
\newcommand{\thrmref}[1]{Theorem~\ref{#1}}

\newcommand{\corolref}[1]{Corollary~\ref{#1}}

\newcommand{\specialcell}[2][c]{%
  \begin{tabular}[#1]{@{}c@{}}#2\end{tabular}}

\makeatletter
\def\blfootnote{\xdef\@thefnmark{}\@footnotetext}
\makeatother

\newtheorem{theorem}{Theorem}[section]

\newtheorem{corollary}[theorem]{Corollary}

\newenvironment{proof}[1][Proof]{\begin{trivlist}
\item[\hskip \labelsep {\bfseries #1}]}{\end{trivlist}}

\newcommand{\qed}{\nobreak \ifvmode \relax \else
      \ifdim\lastskip<1.5em \hskip-\lastskip
      \hskip1.5em plus0em minus0.5em \fi \nobreak
      \vrule height0.75em width0.5em depth0.25em\fi}

\def\BibTeX{{\rm B\kern-.05em{\sc i\kern-.025em b}\kern-.08em
    t\kern-.1667em\lower.7ex\hbox{E}\kern-.125emX}}

\usepackage{newlfont}
\hypersetup{pageanchor=false}
\begin{document}
\title{An Approximately Optimal Algorithm for Scheduling Phasor Data Transmissions in Smart Grid Networks}
\author{\IEEEauthorblockN{K. G. Nagananda and Pramod Khargonekar, \emph{Fellow}, \emph{IEEE}}\thanks{K. G. Nagananda is with People's Education Society (PES) University, Bangalore 560085, INDIA, E-mail: \texttt{kgnagananda@pes.edu}. Pramod Khargonekar is with the University of Florida, Gainesville, FL 32611, U.S.A. E-mail: \texttt{ppk@ece.ufl.edu}}}

\pagenumbering{gobble}
\date{}
\setlength{\droptitle}{-0.5in}
\maketitle

\begin{abstract}
In this paper, we devise a scheduling algorithm for ordering transmission of synchrophasor data from the substation to the control center in as short a time frame as possible, within the realtime hierarchical communications infrastructure in the electric grid. The problem is cast in the framework of the classic job scheduling with precedence constraints. The optimization setup comprises the number of phasor measurement units (PMUs) to be installed on the grid, a weight associated with each PMU, processing time at the control center for the PMUs, and precedence constraints between the PMUs. The solution to the PMU placement problem yields the optimum number of PMUs to be installed on the grid, while the processing times are picked uniformly at random from a predefined set. The weight associated with each PMU and the precedence constraints are both assumed known. The scheduling problem is provably NP-hard, so we resort to approximation algorithms which provide solutions that are suboptimal yet possessing polynomial time complexity. A lower bound on the optimal schedule is derived using branch and bound techniques, and its performance evaluated using standard IEEE test bus systems. The scheduling policy is power grid-centric, since it takes into account the electrical properties of the network under consideration.
\end{abstract}
%\vspace{-0.1in}
\vspace{-0.25in}
\section{Introduction}\label{sec:introduction}
In a typical electric power grid, the communications infrastructure is centered around communications between the individual substations and control centers. In the existing framework, this communications structure has disadvantages in that it offers slow automatic control by the control center and even slower manual control by system operators \cite{Hauser2005}. To overcome such drawbacks, installing GPS satellite-synchronized PMUs \cite{Phadke2008} across the grid which can record phasor data at very high frequencies are proposed. As per the year 2014, 1100 PMUs had been installed across the U.S. Eastern Interconnect offering substantial coverage of the transmission system \cite{2014}, while China has been following large-scale implementation of PMUs for wide area monitoring \cite{Yang2007}.

The high frequency of operation of PMUs results in accumulation of the voluminous phasor data at the substations for further processing and transmission. This is especially true with the proposition by the electric utilities to install large numbers of PMUs, resulting in overwhelming the existing communications network on the grid. The problem was highlighted in a vision paper on the transmissions infrastructure for the future grid \cite{Bose2010}, where a hierarchical communications architecture for realtime wide area monitoring and control was examined. Basically, the following architecture was envisioned: Phasor data at the substations form the first layer of information. Grid status routers form the interfacing layer between the substations and control centers, which form the higher level of information. Phasor data is transmitted from the substations to control centers via status routers over dedicated communication channels (see \cite[Section I B]{Bose2010}). This architecture is preferred over the more traditional redundant communications architecture (see \cite[Fig. 3-3]{2010b}), which presents problems related to management of phasor data especially with the increasing numbers of PMUs on the grid.

In this paper, we base our study within the realtime hierarchical communications infrastructure. Essentially, we concern ourselves with the following problem: Devise an algorithm to schedule transmission of phasor data (collected from a network of PMUs) from the substation to the control center in as short a time frame as possible, given the hierarchical information-flow architecture proposed in \cite{Bose2010}. An overview of the problem setup is presented in the following paragraph.

We consider a power network with $B$ buses and $K$ branches. There are $N < B$  PMUs  in  the grid, which could, for example, arise from optimal PMU placement algorithms  \cite{Baldwin1993} - \nocite{Xu2004} \nocite{Nuqui2005}\nocite{Gou2008}\nocite{Gou2008a}\nocite{Zhang2010}\nocite{Kekatos2012}\cite{Li2013}. Time is divided into frames with the duration of each frame equal to $t$ units.
\begin{figure}[t]
\centering
\includegraphics[height=1.25in,width=3.5in]{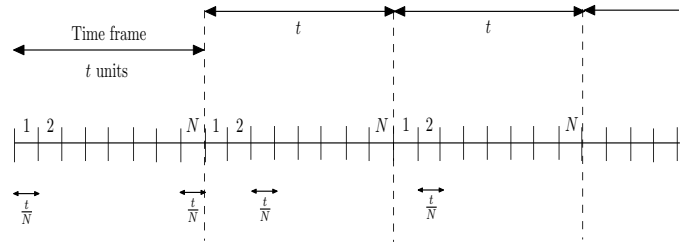}
\caption{Time division multiplexing of PMU transmission.}
\label{fig:pmu_trans1}
\end{figure}
A time frame is further divided into $N$ slots each of duration $\frac{t}{N}$ time units, as illustrated in \figref{fig:pmu_trans1}. Within a time frame, the substation (which has gathered phasor data from the $N$ PMUs) transmits to the control center via the grid status routers on dedicated channels of finite capacity. Given this setup, a transmission schedule for phasor data from the substation is proposed, so that the control center can use the \emph{ordered} set of phasor data to more quickly and more reliably determine changes in the system state. Specifically, referring to \figref{fig:pmu_trans1}, the fundamental question that we ask in this paper is the following: Would it be possible to minimize the processing time of phasor data from the $N$ PMUs across the grid, {\ie}, to devise a scheduling policy that allows phasor data from the $N$ PMUs to be transmitted from the substation in the specified order in \emph{as small a time frame $t$ as possible}?

The question posed in the previous paragraph is answered by casting the scheduling problem in the framework of single machine job scheduling with precedence constraints \cite{Pinedo2012}. In simple words, we come up with a schedule to process phasor data from the optimal $N$ PMUs in as small a time frame $t$ as possible, under the precedence constraints between the PMUs; the precedence constraints are imposed to quantify the importance or significance (in terms of electrical connectedness) of a PMU compared to other PMUs in the network. Since the problem is provably NP-hard, we address it via approximation algorithms, which provide solutions that are suboptimal yet possessing polynomial-time complexity; this component is the main contribution of the paper. The scheduling analysis presented herein is applicable without modifications to any configuration of PMU placement with precedence constraints of any form.

We now discuss the role of precedence constraints in the scheduling problem. In the hierarchical information flow architecture, a network of PMUs transmit to a central monitoring station, where phasor data is archived for further processing. Employment of PMUs in backup protection schemes has been a popular practice \cite{Jiang2002} - \nocite{Eissa2010}\nocite{Ma2011}\cite{Neyestanaki2015}.  Consider a wide-area monitoring application with a fault detection and isolation scenario. It may be that faults on certain important transmission lines have much larger impact on the power grid than other transmission lines. In this case, it is important that the PMU data from the nodes corresponding to these transmission lines be reported to the central monitoring station earlier than data from other nodes. In another scenario, if phasor data from a particular PMU is found to be anomalous (transmission lines breaking, unexpected power demand trend, cyber attacks, unauthorized remote access to substation databases, {\etc}.), then it might need to be assigned a higher precedence to be transmitted sooner to the control center (for timely monitoring and control) compared to phasor data from other PMUs in the network. In this case, the precedence constraints can change during the operation which is a more aggressive scenario than the one considered in this paper. The main point is that it is quite likely that application requirements will lead to precedence constraints on data transmission from various PMUs in the grid. In this paper, for analytical purposes, the precedence constraints are assumed known.

In \cite{Nagananda2014a}, the first author and his colleagues devised a PMU scheduling scheme for transmission of phasor data from a network of $N$ PMUs to the control center without a specific communications architecture, as taken into account in the present paper. The objective there was to improve the performance of a fault detection scheme using scheduling, by accumulating phasor data from ``all'' the $N$ PMUs to execute the fault detection procedure. This was clearly undesirable for realtime processing, since the scheme incurred a delay of $(N-1)t/N$ time units per frame for each PMU which can trigger a hazardous situation even for moderate values of $N$ and $t$. In this paper, we are motivated by realtime aspect of the hierarchical information-flow architecture proposed in \cite{Bose2010}. The problem considered in this paper is especially relevant when it is required to gather and analyze synchrophasor data for real-time decision making and wide area monitoring.

Communications scheduling in power networks has been addressed from various angles. On the one hand, there are algorithms for power scheduling, while on the other hand, scheduling policies have been devised for multimedia data, user-access for smart power appliances, {\etc}. For instance, a power scheduling scheme for the smart grid was proposed in \cite{Zhou2012} to improve the quality of experience (QoE); the QoE represents the customers' degree of satisfaction. A multi-time scheduling scheme was presented in \cite{He2013} in the framework of Markov decision processes for traditional and opportunistic energy users, under supply uncertainty due to variable and non-stationary wind generation, demand uncertainty owing to the stochastic behavior of a large number of opportunistic users and the coupling between sequential decisions across multiple timescales. A scheduling scheme for smart grid traffic (control commands, multimedia sensing data and meter readings) in a priority-based manner, aided by wireless sensor networks and cognitive radio technology, was presented in \cite{Huang2013}. In \cite{Chen2013}, a joint access and scheduling policy was proposed for in-home appliances (schedulable and critical) to coordinate the power usage to keep the total energy demand for the home below a target value.

Contrasting to work in the existing literature, the focus of this paper is scheduling, or ordering, transmissions of phasor data from the substations to the control center on the power network by devising a feasible schedule in as small a time frame as possible. We are not concerned with power and grid-traffic scheduling that have been widely addressed by the smart grid community. To the best of our knowledge, this work and the one presented in \cite{Nagananda2014a} are the first results on phasor data scheduling reported in the literature.

The remainder of the paper is organized as follows. In \secref{sec:problem_statement}, a mathematically precise problem statement is presented. \secref{sec:solution_methodology} comprises the procedure to compute the weights associated with the PMUs, and the details of the approximation strategy to address the scheduling problem. The scheduling scheme is presented in \secref{sec:scheduling_scheme}. An illustrative example along with the experimental results are provided in \secref{sec:experimental_results}. Conclusion and possible avenues for future research are provided in \secref{sec:conclusion}.

\vspace{-0.1in}
\section{Problem statement}\label{sec:problem_statement}
Consider the set $\mathcal{N} = \{1,\dots,N\}$ of $N$ PMUs installed on the power grid. The network of PMUs transmit phasor data to the substation where it will be archived for local processing and global operations (such as wide area monitoring, {\etc}.) \cite{Bose2010}. Phasor data from each PMU $n \in \mathcal{N}$ has a positive processing time, $p_n > 0$ and a nonnegative weight $w_n \geq 0$; for purpose of analysis, the weights are assumed to be known. Let $C_n$ denote the time at which phasor data from the $n^{\text{th}}$ PMU is processed completely by the control center in a feasible schedule. The scheduling scheme is required to satisfy precedence constraints, which are specified in the form of a directed acyclic graph $\mathcal{G} = (\mathcal{N}, \mathcal{P})$ where $(i, j) \in \mathcal{P}$ implies that the phasor data of $\text{PMU}_i$ must be completely processed by the control center before the substation can transmit data of $\text{PMU}_j$. The main objective is to find a feasible schedule for the substation to transmit phasor data so as to minimize the weighted sum of the processing times, which is given by
\begin{eqnarray}
\sum_{n\in \mathcal{N}}w_nC_n.
\label{eq:weighted_protime}
\end{eqnarray}
Preemption and control-center idle time are forbidden and hence phasor data all the PMUs should be processed in the interval $[0, t]$, where $t = \sum_{n \in \mathcal N} p_n$. Preemption refers to a temporary interruption of processing, while idle time refers to the lapse between the end-time of processing phasor data of $\text{PMU}_i$ and the start-time of processing data of $\text{PMU}_j$.

In the classic setting, the aforementioned problem falls in the framework of single-machine job scheduling with precedence constraints \cite{Pinedo2012}.  The machine scheduling problem was shown to be strongly NP-hard \cite{Lawler1978}, \cite{Schuurman1999}, and has received considerable attention both from a theoretical viewpoint \cite{Chekuri1999} - \nocite{Correa2005}\cite{Adolphson1977} and for devising practical strategies \cite{Ibaraki1994} - \nocite{Chudak1999}\cite{Potts1980}. Given the difficulty in devising efficient polynomial-time algorithms for optimally solving NP-hard problems, it is the usual practice to resort to polynomial-time suboptimal schemes, referred to as approximation algorithms \cite{Vazirani2003}. The approximation is optimal up to a small constant factor (for example, within 2\% of the optimal solution). An $\alpha$-approximation algorithm runs in polynomial time and produces for every instance a feasible schedule of cost at most $\alpha$ times that of an optimal schedule. The value $\alpha$ is called the performance guarantee, or integrality gap, of the algorithm. With the increasing demand for installing a very large number of PMUs on the power grid, and given the societal importance of the problem, we will develop power grid-centric approximation algorithms to efficiently solve the scheduling problem posed in \eqref{eq:weighted_protime}.

\section{Solution methodology}\label{sec:solution_methodology} %\vspace{-0.05in}
In this section, we elaborate on the approximation strategy to solve the scheduling problem posed in the optimization setup \eqref{eq:weighted_protime}. which is the subject topic of this paper. As mentioned in the earlier sections, we assume knowledge of the weights $w_n$, $n \in \mathcal{N}$, and precedence constraints between the PMUs.

We first present a 0-1 programming formulation of the optimization setup \eqref{eq:weighted_protime}. We define a variable $\delta_{nm}$ as follows:
\begin{eqnarray}
\delta_{nm} = \begin{cases}
                1,~\text{if}~\text{PMU}_n~\text{is processed before}~\text{PMU}_m \\
                0,~\text{otherwise}.
              \end{cases}
\label{eq:delta_define}
\end{eqnarray}
Let $\gamma_{nm} = 1$ when the precedence constraints specify that $\text{PMU}_n$ is a predecessor of $\text{PMU}_m$ and let $\gamma_{nm} = 0$ otherwise. The processing of $\text{PMU}_m$ occurs at time $\sum_{n}p_n\delta_{nm} + p_m$. Therefore, \eqref{eq:weighted_protime} can be written as
\begin{eqnarray}
\min \sum_n\sum_m p_n\delta_{nm}w_m &+& \sum_m p_m w_m \label{eq:lp_relax1} \\
\text{subject to}~ \delta_{nm} &\geq& \gamma_{nm} \label{eq:lp_relax_constraint1} \\
\delta_{nm} + \delta_{mn} &=&  1 \label{eq:lp_relax_constraint2} \\
\delta_{nm} + \delta_{mk} + \delta_{kn} &\geq& 1 \label{eq:lp_relax_constraint3}\\
\delta_{nm} &\in& \{0, 1\} \label{eq:lp_relax_constraint4} \\
\delta_{nn} &=& 0, \label{eq:lp_relax_constraint5}
\end{eqnarray}
$n = 1,\dots,N$, $m = 1,\dots,N$, $k = 1,\dots,N$, $n \neq m$, $n \neq k$, $m \neq k$. The constraint \eqref{eq:lp_relax_constraint1} specifies that $\delta_{nm} = 1$ whenever $\text{PMU}_n$ is a predecessor of $\text{PMU}_m$, while \eqref{eq:lp_relax_constraint2} specifies that $\text{PMU}_n$ is to be scheduled either before or after $\text{PMU}_m$. The matrix $\bm{\delta} = (\delta_{nm})$ is the adjacency matrix of a complete directed graph $\mathcal{G}_{\bm\delta}$. We introduce a constraint $\delta_{mn} + \delta_{nk} + \delta_{km} \leq 2$ to ensure that the graph $\mathcal{G}_{\bm\delta}$ contains no cycles. The constraint \eqref{eq:lp_relax_constraint2} implies that the aforementioned constraint is equivalent to \eqref{eq:lp_relax_constraint3}, which, therefore, ensures that $\mathcal{G}_{\bm\delta}$ contains no cycles. When all the constraints are satisfied, $\mathcal{G}_{\bm\delta}$ defines a complete ordering of the PMUs. The coefficient $p_nw_m$ or $\delta_{nm}$ in \eqref{eq:lp_relax1} denotes the cost of scheduling $\text{PMU}_n$ before $\text{PMU}_m$. The cost of scheduling $\text{PMU}_n$ before $\text{PMU}_m$ is denoted in terms of the cost matrix $\bm C \triangleq (c_{nm})$, where
\begin{eqnarray}
c_{nm} = \begin{cases}
             p_nw_m, ~\text{if}~\gamma_{mn} = 0, \\
             \infty, ~\text{if}~\gamma_{mn} = 1.
         \end{cases}
\label{eq:cost_matrix}
\end{eqnarray}
Whenever the precedence constraints specify that $\text{PMU}_m$ is a predecessor of $\text{PMU}_n$, we have $c_{nm} = \infty$. The setup \eqref{eq:lp_relax1} can now be written as
\begin{eqnarray}
\min \sum_n\sum_m c_{nm}\delta_{nm} + \sum_m p_m w_m,
\end{eqnarray}
subject to constraints specified by \eqref{eq:lp_relax_constraint2}, \eqref{eq:lp_relax_constraint3}, \eqref{eq:lp_relax_constraint4} and \eqref{eq:lp_relax_constraint5}. Each of the constraints in \eqref{eq:lp_relax_constraint2} is incorporated into a Lagrangian function with multipliers $\alpha_{nm}$ as follows
\begin{eqnarray}
\nonumber L^{(0)} &=& \sum_n\sum_m \left(c_{nm} - \alpha_{nm} - \alpha_{mn} \right)\delta_{nm}\\ && + \sum_{n}\sum_{m \neq n}\alpha_{nm} + \sum_m p_m w_m,
\label{eq:lagrange1}
\end{eqnarray}
which is to be minimized subject to \eqref{eq:lp_relax_constraint2}, \eqref{eq:lp_relax_constraint3}, \eqref{eq:lp_relax_constraint4} and \eqref{eq:lp_relax_constraint5}. We let
\begin{eqnarray}
\alpha_{nm} = \alpha_{mn} = \frac{1}{2}\min\{c_{nm}, c_{mn}\}, n \neq m,
\label{eq:setting_alpha}
\end{eqnarray}
so that the multipliers provide as large a contribution as possible to the lower bound. With this, we now define a reduced cost matrix $\bm{C}^{(0)} \triangleq (c_{nm}^{(0)})$, where $c_{nm}^{(0)} = c_{nm} - \alpha_{nm} - \alpha_{mn}$.

In the presence of cycles, constraints \eqref{eq:lp_relax_constraint2} and \eqref{eq:lp_relax_constraint3} can be used to derive cycle elimination constraints involving $q$ edges of the form
\begin{eqnarray}
\sum_{i=1}^{q}\delta_{n_i,n_{i+1}} \geq 1,
\label{eq:cycle_eliminate}
\end{eqnarray}
where $n_1 = n_{q+1}$ and $n_1,\dots,n_q$ correspond to $q$ different PMUs. Suppose that a Lagrangian relaxation of $r - 1$ of the constraints \eqref{eq:cycle_eliminate} are performed using multipliers $\beta^{(1)},\dots,\beta^{(r-1)}$ giving the following Lagrangian function
\begin{eqnarray}
\nonumber L^{(r-1)} &=& \sum_{n}\sum_{m}c_{nm}^{(r-1)}\delta_{nm} + \sum_{n}\sum_{m\neq n} \alpha_{nm} \\ && + \sum_{k=1}^{r-1}\beta^{(k)} + \sum_m p_m w_m.
\label{eq:lagrange_cycle1}
\end{eqnarray}
The constraint \eqref{eq:cycle_eliminate} is introduced into \eqref{eq:lagrange_cycle1} to get
\begin{eqnarray}
L^{(r)} &=& L^{(r-1)} + \beta^{(r)}\left(1 - \sum_{i=1}^{q}\delta_{n_i,n_{i+1}} \right), \label{eq:lagrange_cycle2} \\
\nonumber \beta^{(r)} &=& \min_{i \in \{1,\dots,q\}} \left\{c^{(r-1)}_{n_i, n_{i+1}} \right\}.
\end{eqnarray}
Choosing $\beta^{(r)}$ as large as possible retains the nonnegativity of the coefficients of the variables. The update equation for the reduced cost matrix is $\bm{C}^{(r)} = \bm{C}^{(r-1)} - \beta^{(r)}\bm{B}^{(r)}$, where the matrix $\bm{B}^{(r)} \triangleq (b_{nm}^{(r)})$ is defined as follows:
\begin{eqnarray}
b_{nm}^{(r)} = \begin{cases}
                1,~\text{if}~ n = n_i~\text{and}~ m = n_{i+1},\\
                0,~\text{otherwise},
               \end{cases}
\label{eq:update_matrix1}
\end{eqnarray}
where $i = 1,\dots,q$. Thus, we can write $L^{(r)}$ in the same form as $L^{(r-1)}$ as given by \eqref{eq:lagrange_cycle1}. Using $L^{(r)}$, a lower bound on the optimal solution of \eqref{eq:weighted_protime} can be obtained.

\begin{theorem}
The following is a lower bound for the sum of weighted completion times:
\begin{eqnarray}
\text{LB}^{(r)} = \sum_{n}\sum_{m \neq n}\alpha_{nm} + \sum_{i=1}^{r}\beta^{(i)} + \sum_{m}p_m w_m.
\label{eq:lower_bound}
\end{eqnarray}
\label{thm:lower_bound}
\end{theorem}\vspace{-0.1in}
\begin{proof}
From the aforementioned development, we see that $c_{nm}^{(r)} \geq 0$, $n, m = 1,\dots,N$ and $n \neq m$. Therefore, $L^{(r)}$ given by \eqref{eq:lagrange_cycle2} can be minimized by setting $\delta_{nm} = 0$ whenever $c_{nm}^{(r)} >  0$ which implies $c_{nm}^{(r)}\delta_{nm} =  0$. This proves \thrmref{thm:lower_bound}.
\end{proof}

From \eqref{eq:lower_bound}, it is clear that the lower bound comprises only cycle elimination constraints. Furthermore, increasing the lower bound amounts to introducing more such constraints which is equivalent to increasing the multipliers. In other words, to obtain the tightest lower bound, as many cycle elimination constraints should be introduced. The following theorem states that when no more such constraints can be found, the reduced cost matrix $\bm{C}^{(r)}$ can be used to find a feasible schedule of the PMUs.

\begin{theorem}
There exist variables satisfying the constraints \eqref{eq:lp_relax_constraint1} - \eqref{eq:lp_relax_constraint5} such that $\sum_n\sum_m c_{nm}^{(r)}\delta_{nm} = 0$ if and only if no additional constraint with a positive multiplier can be introduced into the Lagrangian function $L^{(r)}$ given by \eqref{eq:lagrange_cycle2}.
\label{thm:no_constraints}
\end{theorem}
\begin{proof}
Consider a directed graph $\mathcal{G}$, where a vertex denotes a PMU, and if $c_{nm}^{(r)} > 0$ then there exists an edge between the vertex pair $(n, m)$. The variables satisfying \eqref{eq:lp_relax_constraint1} - \eqref{eq:lp_relax_constraint5} and $\sum_n\sum_m c_{nm}^{(r)}\delta_{nm} = 0$ define an ordering of the PMUs for which the order graph $\mathcal{G}_{\bm{\delta}}$ contains no edge of $\mathcal{G}$. In other words, the graph obtained by reversing all edges of $\mathcal{G}_{\bm{\delta}}$ contains all edges of $\mathcal{G}$. This implies $\mathcal{G}$ contains a partial order without cycles; therefore, no cycle elimination constraints with a positive multiplier exists. On the other hand, if no cycles exist in $\mathcal{G}$, then there exists an ordering consistent with the partial ordering in $\mathcal{G}$. The reverse ordering provides the values of the variables which satisfy \eqref{eq:lp_relax_constraint1} - \eqref{eq:lp_relax_constraint5} such that $\sum_n\sum_m c_{nm}^{(r)}\delta_{nm} = 0$. This proves \thrmref{thm:no_constraints}.
\end{proof}

If a schedule of PMUs can be found from $\bm{C}^{(r)}$ by defining values of the variables that satisfy the constraints \eqref{eq:lp_relax_constraint1} - \eqref{eq:lp_relax_constraint5}, then the sum of the weighted completion times for this schedule is an upper bound $\text{UB}^{(r)}$ computed as follows:
\begin{eqnarray}
\text{UB}^{(r)} = \text{LB}^{(r)} + \sum_{i=1}^{r}\beta^{(r)}\left(\sum_n\sum_m b_{nm}^{(i)}\delta_{nm} - 1\right),
\end{eqnarray}
where the matrix $\bm{B}^{(i)}$, $i = 1,\dots,r$ is defined by \eqref{eq:update_matrix1}.

\begin{corollary}
$\text{LB}^{(r)} = \text{UB}^{(r)}$ if and only if all constraints in the Lagrangian function are satisfied with equality.
\label{coro:lower_upper}
\end{corollary}
From \corolref{coro:lower_upper}, we see that the lower bound can be made tighter by using the constraints that are satisfied as inequalities. Let the PMUs be renumbered so that the schedule obtained from $\bm{C}^{(r)}$ is $(1,\dots,N)$, and therefore, $\bm{C}^{(r)}$ is a lower triangular matrix. Now, suppose that a constraint with multiplier $\beta$ satisfying the strict inequality has been found. If this $\beta$ is removed from the Lagrangian function, the entries of the matrix $\bm{C}^{(r)}$ increase by a factor of $\beta$. Thus, for $n < m$, each $\delta_{nm}$ produces an element $c_{nm}^{(r)} = \beta$ placed above the leading diagonal in $\bm{C}^{(r)}$. New cycle elimination constraints can be introduced into the Lagrangian function by using such $\delta_{nm}$ for which $c_{nm}^{(r)} > 0$ along with $\delta_{ij}$, $n \leq j < i \leq m$ for which $c_{ij}^{(r)} > 0$. If new cycle elimination constraints having multipliers which sum to $0 < \beta^{\ast} \leq \beta$ are introduced into the Lagrangian function for each of $\eta \geq 2$ variables $\delta_{nm}$ of the original constraint, then the original chosen constraint is reintroduced into the Lagrangian function with a reduced multiplier $\beta - \beta^{\ast}$, and the lower bound increases by $(\eta - 1)\beta^{\ast}$. By choosing $\eta$ and $\beta^{\ast}$ as large as possible, we can provide the maximum increment to the lower bound. Note that, the new Lagrangian function necessitates finding a new schedule such that $\bm{C}^{(r)}$ is a lower triangular matrix.

To find new constraints involving the variable $\delta_{nm}$ $(n < m)$ mentioned in the aforementioned paragraph, we present the following simple procedure: Consider a network of buses $n,\dots,m$, where $m$ and $n$ denote the source and sink, respectively. An edge exists between buses $i$ and $j$ with capacity $c_{ij}^{(r)}$ whenever $n \leq j < i \leq m$. A capacity of $\beta$ is placed on the source bus. The problem of finding the maximum flow from the source to sink buses is equivalent to that of generating new constraints. The maximum flow can be decomposed into different flows along paths from source to sink buses. The variables corresponding to the edges in a path from source to sink can be combined with $\delta_{nm}$ to produce cycle elimination constraints; the associated multiplier is given by the flow along that path. When $\bm{C}^{(r)}$ is reduced, the entries remain nonnegative, since the flow along any edge cannot exceed its capacity. Also, the lower bound is the best possible for the constraints obtained from the network flow problem. This is because, all constraints contain at least one variable corresponding to an edge which cuts across the minimum cut set, and the fact that the maximum flow is equal to the value of the minimum cut set.
\vspace{-0.1in}
\subsection{Method to achieve the lower bound}\label{subsec:method_lowerbound}
The first step is to compute $L^{(0)}$ given by \eqref{eq:lagrange1} which requires $\mathcal{O}(N^2)$ computations. The second step is to decide the cycle removing constraints and the order in which they are to be introduced into the Lagrangian function.

For the bound to be tight, from \corolref{coro:lower_upper} we know that any constraint that is added should be satisfied with equality, and is a suggestive rule indicating which constraint should be used. Given this condition, for small values of $q$  inequality \eqref{eq:cycle_eliminate} is satisfied with equality. Also, a schedule generated by a heuristic procedure indicates whether the constraint is satisfied as an equality. With this setup, we add all cycle elimination constraints with three edges in $\mathcal{O}(N^3)$ steps that are satisfied as equalities when the values of the variables correspond to heuristic schedule. This process is repeated for constraints with four edges in $\mathcal{O}(N^4)$ steps.

Lastly, a technique is needed to remove remaining cycles and to enable a schedule to be found from the reduced cost matrix $\bm{C}^{(r)}$. It is easy to see that any PMU $n$ with $c_{nm}^{(r)} = 0$ for all PMUs $m$ can be placed in the first unfilled position in the schedule, and row $n$ and column $n$ can be deleted from the reduced cost matrix. If there is a choice of PMUs that can be scheduled, the one that is scheduled first in the heuristic method is chosen. The process if repeated until no PMU can be scheduled in the first unfilled position of the schedule. In the event that a complete schedule has not been generated, it is possible to schedule PMUs in the final position of the schedule in a similar manner.

From \thrmref{thm:no_constraints}, it follows that when unscheduled PMUs remain, at least one additional constraint can be included in the Lagrangian function to tighten the lower bound. We pick a $\text{PMU}_n$ such that the number $m$ of unscheduled PMUs with $c_{nm}^{(r)} > 0$ is as small as possible. A search procedure involving $\mathcal{O}(N^2)$ steps can be invoked if a cycle elimination constraint with $\text{PMU}_n$ exists. In the absence of a cycle, another unscheduled PMU is similarly examined, and the process is repeated till all the PMUs are scheduled.

During implementation of the lower bound, we saw that $\mathcal{O}(N^4)$ steps are required for the computation of the lower bound. For the elimination of the cycles, the network flow problem, whose solution requires $\mathcal{O}(N^3)$ steps, is solved for each variable taking the value one in each constraint satisfied as a strict inequality. There are $\mathcal{O}(N^2)$ constraints each comprising $\mathcal{O}(N)$ variables. Therefore, we need $\mathcal{O}(N^6)$ steps if we are not forced to generate a new schedule during the network flow problem. Thus, it is clear that the number of steps required by the network flow is pseudopolynomial, provided a new schedule is generated so that the lower bound keeps increasing by one unit.
\vspace{-0.05in}
\section{Scheduling scheme}\label{sec:scheduling_scheme}
In this section, we describe the heuristic method used before the branch and bound algorithm, along with the branching rule and the search procedure.

%\subsection{Heuristics}\label{subsec:heuristics}
We employ a simple heuristic method to obtain a schedule. Heuristics are used to indicate which constraints are to be used in the lower bound, and are not concerned with providing upper bounds to the optimization setup. Basically, we select a schedule at random that satisfies all the precedence constraints. To improve the heuristic, a group of PMUs is selected which consists of PMUs that are scheduled in consecutive positions in the given schedule such that there is a precedence constraint between each adjacent pair of PMUs. The PMUs of this group are removed from their current positions in the schedule and inserted elsewhere such that the new schedule satisfies all the precedence constraints in addition to having a smaller sum of weighted completion times than the previous sequence. The process is repeated until no further improvements are possible.

%\subsection{Branching rule}\label{subsec:branching_rule}
The central idea behind the branching rule is to reduce the difference between the lower and upper bounds computed at the bus from which we decide to branch. We search for constraints which satisfy the inequalities and the one which has the largest multiplier is selected. The variable $\delta_{nm}$ which appears in this constraint is selected so that $\delta_{nm}$ takes the value one, and such that $\text{PMU}_n$ is not a predecessor or a successor of $\text{PMU}_m$. The search tree with two branches are formed, such that in one branch $\text{PMU}_m$ is constrained to be scheduled before $\text{PMU}_n$ and in the other branch $\text{PMU}_n$ is constrained to be scheduled before $\text{PMU}_m$. The precedence constraints between other pairs of PMUs are directly implied by transitivity. The variable $\delta_{ij}$ is removed from the Lagrangian without altering the lower bound if $\text{PMU}_i$ is scheduled before $\text{PMU}_j$. In the first branch of the search tree, where $\text{PMU}_n$ is constrained to be scheduled after $\text{PMU}_m$, new schedules will emerge after the addition of more cycle elimination constraints. In the second branch, where $\text{PMU}_n$ is constrained to be scheduled before $\text{PMU}_m$, the constraint from which the variable $\delta_{nm}$ is originally found will be removed.

%\subsection{Search procedure}\label{subsec:search_procedure}
The goal of the search strategy is to specify from which bus to branch. A simple active node search is employed which uses backtracking to compute the reduced cost matrix from $\bm{C}^{(0)}$ and the list of constraints. Doing so, we avoid storing the reduced cost matrix at node of the search tree. In the following section, we perform computer simulations to validate the theoretical findings of this paper.

\vspace{-0.1in}
\section{Illustrative Example}\label{sec:experimental_results}
The algorithm was tested on different grids. For each grid, the optimal number of PMUs was obtained by solving the PMU placement problem, with the objective to achieve complete network observability ignoring zero injection measurements. For phasor data from each PMU, an integer processing time (in ms) from the uniform distribution [1, 50] is chosen. The algorithm was implemented using MATLAB$^\circledR$ on a Windows 7 PC, with a dual-core central processor, and 4.00 GB random access memory.

In the theoretical analysis carried out in the earlier sections of the paper, we had assumed knowledge of the weights and the precedence constraints. For the purpose of the experimental study presented in this section, we will utilize a simple procedures to derive the weight $w_n$, $n = 1,\dots,N$ of each PMU and the precedence constraints between the PMUs. This is being done for this illustrative example. As noted earlier, we expect that the precedence constraints will come naturally from the power grid application requirements.

The SVD of the bus admittance matrix is a useful measure of the electrical connectedness between various components in the power network. A thorough exposition on this topic was first reported in \cite{Wang2010}, where the authors introduce the term ``electrical centrality measure'' to quantify the degree of connectedness between buses in the grid. It was analytically demonstrated in \cite{Wang2010} that typical grid admittance matrices have singular values and vectors with only a small number of strong components, and performing SVD on the admittance matrix leads to no loss of connectivity-information.

The weights $w_n$, $n \in \mathcal{N}$ are obtained as follows:
\begin{enumerate}[1)]
    \item We obtain the optimal number $(N)$ of PMUs by solving the PMU placement problem. \label{step1}
    \item In the $B\times B$ bus admittance matrix, we choose the rows and columns which correspond to the bus numbers where PMUs are installed, leading to an $N \times N$ sub-matrix. \label{step2}
    \item We perform the SVD of the $N \times N$ sub-matrix to obtain the singular values and singular vectors. The $N \times 1$ left and right singular vectors are denoted $\bm{u}_n$ and $\bm{v}_n$, respectively, while the singular values are denoted $\sigma_n$. \label{step3}
    \item We compute the Euclidian norm of the vectors $\sigma_n\bm{u}_{n}$, and define the weight associated with the PMUs as follows:
        \begin{eqnarray}
        && w_n \triangleq \left\lceil\frac{||\sigma_n\bm{u}_{n}||}{N}\right\rceil,~~ n = 1, \dots, N.
        \label{eq:pmu_weight}
        \end{eqnarray}
        Note that, the index of each entry of the vector $\sigma_n\bm{u}_{n}$ corresponds to a bus location where a PMU is installed. \label{step5}
\end{enumerate}
\vspace{-0.1in}
The precedence constraints are obtained as follows:
\begin{enumerate}[1)]
\item In the vector $\sigma_1\bm{u}_{1}$, {\ie}, the first column of the $N \times N$ sub-matrix obtained as above, the PMU placed on the entry with the highest magnitude is given the highest precedence. $\bm{u}_{1}$ is the eigenvector corresponding to the largest eigenvalue.
        \label{transmit_first}
\item The procedure in Step \ref{transmit_first} is repeated for the remaining vectors $\sigma_n\bm{u}_{n}$, $n = 2,\dots,N$, where the $\bm{u}_{n}$s are picked in the decreasing order of the corresponding eigenvalues.  \label{conflict}
\end{enumerate}
However, in the aforementioned procedure there is a possibility of conflict. For example, consider two vectors $\sigma_1\bm{u}_{1}$ and $\sigma_3\bm{u}_{3}$. Suppose the entry having the largest magnitude in vector $\sigma_1\bm{u}_{1}$ is the same as the entry having the largest magnitude in vector $\sigma_3\bm{u}_{3}$. Then, the procedure picks the same PMU in both the vectors $\sigma_1\bm{u}_{1}$ and $\sigma_3\bm{u}_{3}$. To resolve this conflict, we propose the following modification: Pick the entry in the vector $\sigma_3\bm{u}_{3}$ having the \emph{second largest magnitude}. If this entry is not the same as the one in vector $\sigma_2\bm{u}_{2}$, then the PMU placed on that entry is given a lower precedence than the PMU obtained in the vector $\sigma_2\bm{u}_{2}$. This simple procedure is implemented for all the vectors $\sigma_n\bm{u}_{n}$.
%Note that, the precedence given to PMU transmissions is solely based on the electrical connectedness of buses in the network, making it different from scheduling schemes devised for typical communications networks. The following illustration provides more clarity on setting the precedence constraints.

We consider the IEEE 14-bus system to illustrate the procedure adopted for establishing the precedence constraints between the PMUs. We obtain the optimal number PMUs employing the topology of the grid \cite{Gou2008}, \cite{Gou2008a}. We consider a single time frame, and implement the following steps:
\begin{figure}[h]
\centering
\includegraphics[height=2.5in,width=3.25in]{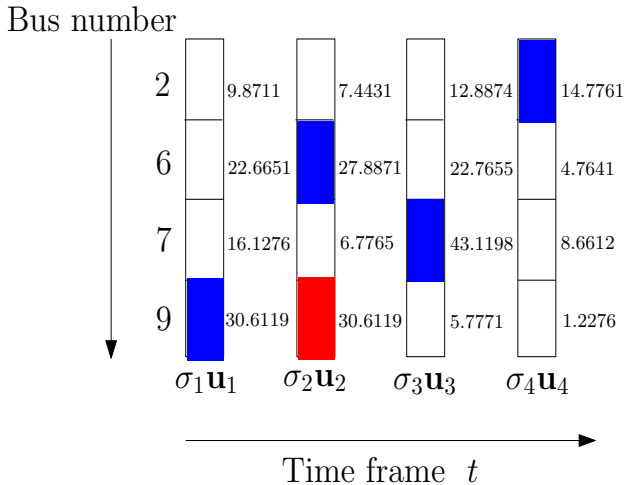}
\caption{Precedence between PMUs for the IEEE 14-bus network. The values of the elements of the vectors $\sigma_n\bm{u}_{n}$ are also indicated.}
\label{fig:pmu_precedence}
\end{figure}
\begin{enumerate}[(1)]
\item For the 14-bus network, an optimum of $N$ = 4 PMUs is to be placed on buses numbered 2, 6, 7 and 9 for complete network observability ignoring zero injection measurements (see \cite{Xu2004}, \cite{Gou2008}, \cite{Gou2008a}).

\item In the 14 $\times$ 14 bus admittance matrix, pick the rows and columns numbered 2, 6, 7 and 9, thereby yielding a 4 $\times$ 4 sub-matrix.

\item Perform the SVD of the 4 $\times$ 4 sub-matrix to obtain the 4 $\times$ 1 right and left singular vectors $\bm{u}_n$ and $\bm{v}_n$, respectively, and the singular values $\sigma_n$, $n$ = 1, \dots, 4. Compute the magnitude of the elements of the vectors $\sigma_n\bm{u}_{n}$. The index of each entry of the vector $\sigma_n\bm{u}_{n}$ corresponds to a bus location where a PMU is installed. The four vectors $\sigma_n\bm{u}_{n}$, $n$ = 1, \dots, 4 are depicted in \figref{fig:pmu_precedence}, where a column denotes a vector, while a box in each column denotes an entry of the vector. The magnitude of the elements of vectors $\sigma_n\bm{u}_{n}$ are also indicated. The number of boxes in each column equals the number $N$ of PMUs installed on the bus system.

\item In the vector $\sigma_1\bm{u}_{1}$ (the first column in \figref{fig:pmu_precedence}), the entry having the largest magnitude appears in the last row - marked in blue. Thus, the PMU placed on bus numbered 9 is given the highest precedence.

\item In the vector $\sigma_2\bm{u}_{2}$ (the second column), the entry having the largest magnitude appears in the last row, similar to that in the vector $\sigma_1\bm{u}_{1}$, again allocating the second highest precedence to the PMU placed on bus numbered 9 (conflict). However, this conflict is resolved by giving the second highest precedence to the PMU placed on the bus numbered 6, which has the second largest magnitude in the vector $\sigma_2\bm{u}_{2}$. In other words, $\text{PMU}_{9} > \text{PMU}_{6}$.

\item Continuing in this fashion, and employing the conflict-resolution strategy, the PMUs installed on buses 2, 6, 7 and 9 are given the precedence in that order, {\ie}, $\text{PMU}_{9} > \text{PMU}_{6} > \text{PMU}_{7} > \text{PMU}_{2}$.
\end{enumerate}
We thus have precedence constraints of the following form between the PMUs: For $n > m$, $\text{PMU}_n > \text{PMU}_m$, {\ie}, $\text{PMU}_n$ is designated to be transmitted before $\text{PMU}_m$.

We consider the IEEE-14, IEEE-30, IEEE-39, IEEE-57, IEEE-118, IEEE-300 and power flow data from two Polish electric grid systems. The bus and line data for the bus networks were obtained using MATPOWER$^\circledR$ \cite{Zimmerman2011}. In \tabref{tab:bus_networks}, we show the different bus networks considered in this paper along with the number of branches and the optimal number $(N)$ of PMUs to be installed on each of the grids for complete network observability ignoring zero injection measurements.
\begin{table}[h!]
\centering
\begin{tabular}{|c|c|c|}
  \hline
  \specialcell{Bus\\ network} & \specialcell{Number of\\branches} & \specialcell{Optimal $(N)$\\ PMUs} \\ \hline
  IEEE-14 & 20 & 4 \\ \hline
  IEEE-30 & 41 & 10 \\ \hline
  IEEE-39 & 46 & 13 \\ \hline
  IEEE-57 & 80 & 17 \\ \hline
  IEEE-118 & 186 & 32 \\ \hline
  IEEE-300 & 411 & 156 \\ \hline
  Polish-2737 & 3506 & 971 \\ \hline
  Polish-3375 & 4161 & 1384 \\
  \hline
\end{tabular}
\caption{Bus networks considered in this paper and the optimal $N$ PMUs obtained by solving the PMU placement problem \eqref{eq:pmuplacement}.}
\label{tab:bus_networks}
\end{table}

We compare the performance of our scheduling scheme with that of a simple greedy algorithm. We begin by considering a directed acyclic graph based on the precedence constraints which are assumed known. Then, we pick that PMU from the directed acyclic graph which has no parent to transmit, and then delete that PMU from the graph. If there are more than one PMU, we pick the one with the largest weight.

In \figref{fig:N_avgtime}, we show the plot of average computation time as a function of the number of PMUs installed on the grid for both the greedy algorithm and for the scheduling scheme devised in this paper. We see that, for less than 200 PMUs, our scheduling scheme has good performance with regard to computation time. However, when more number of PMUs are installed, the average computation time increases significantly. The reason for an increase in the average time for larger number of PMUs is basically computational in nature. Also, the performance of the scheduling scheme is comparable to that of the greedy algorithm for fewer number of PMUs. As the number of PMUs increases, the greedy algorithm behaves very poorly, which corroborates intuition; greedy algorithms approximate a global optimal solution in reasonable time for fewer instances of the problem, however, as the number of instances increases, the performance degrades. A similar behavior is observed in \figref{fig:N_nodes} when we plot the average number of nodes required in the search tree as a function of $N$.
\begin{figure}[h]
\centering
\includegraphics[height=2.5in,width=3.5in]{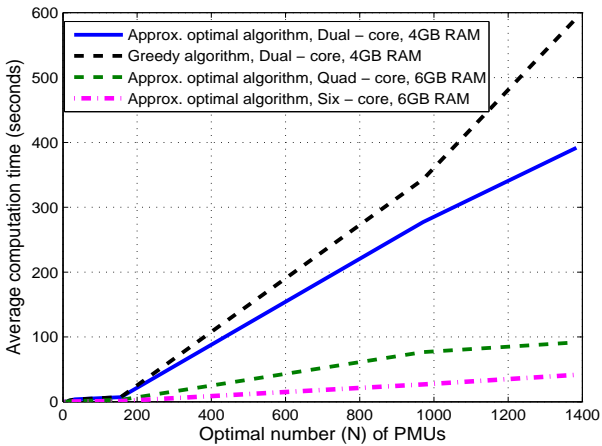}
\caption{Average computation time versus the optimal number $(N)$ of PMUs, using different computational resources.}
\label{fig:N_avgtime}
\end{figure}
\begin{figure}[h]
\centering
\includegraphics[height=2.5in,width=3.5in]{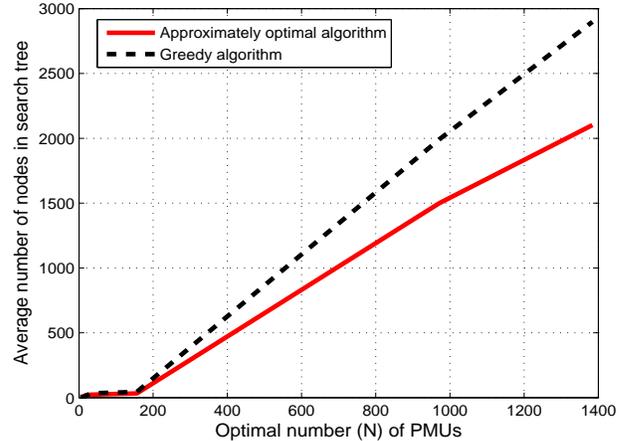}
\caption{Average number of nodes in the search tree versus the optimal number $(N)$ of PMUs.}
\label{fig:N_nodes}
\end{figure}

\begin{table}[h!]
\indent
\begin{tabular}{|c|c|c|c|c|c|}
\cline{3-6}
\multicolumn{1}{c}{} & \multicolumn{1}{c|}{} & \multicolumn{2}{c|}{\specialcell{Max. computation\\ time (sec)}} & \multicolumn{2}{c|}{\specialcell{Avg. computation\\ time (sec)}} \\
\hline
\specialcell{Bus\\ network} & \specialcell{Optimal $(N)$\\ PMUs} & {GA} & {AOA}  & {GA} & {AOA} \\
\hline
IEEE-14 & 4 & 0.182 & 0.312 & 0.061 & 0.091 \\ \hline
IEEE-30 & 10 & 0.371 & 0.598 & 0.112 & 0.132 \\ \hline
IEEE-39 & 13 & 0.612 & 0.861 & 0.158 & 0.198 \\ \hline
IEEE-57 & 17 & 1.125 & 2.761 & 2.118 & 1.41 \\ \hline
IEEE-118 & 32 & 18.112 & 5.871 & 3.547 & 3.87 \\ \hline
IEEE-300 & 156 & 23.129 & 11.651 & 7.242 & 7.182\\ \hline
Polish-2737 & 971 & -- & -- & 343.17 & 277.12 \\ \hline
Polish-3375 & 1384 & -- & -- & 591.71 & 391.65 \\
\hline
\end{tabular}
\caption{Maximum and average computation times. GA - greedy algorithm, AOA - approximately optimal algorithm developed in this paper.}
\label{tab:maximum_computation_time}
\end{table}
In \tabref{tab:maximum_computation_time}, we present the results of maximum and average computation times for the greedy algorithm and the proposed scheduling scheme. Here again, we see that, as the number of PMUs increases, the maximum computation time to obtain a feasible schedule increases, owing to the size of the problem. It is also of interest to note that, for the last two entries in the \tabref{tab:maximum_computation_time}, one cannot obtain the maximum computation time, since the average computation time and the average number of nodes that are calculated and plotted in \figref{fig:N_avgtime} and \figref{fig:N_nodes} represent lower bounds on the respective averages because algorithm does not converge in polynomial time.

To demonstrate the impact of computational resources on the performance of scheduling, we also tested our algorithm on quad-core and six-core processors with 6GB of random access memory. The resulting average computation time for varying $N$ is shown in \figref{fig:N_avgtime}, where we notice reduction in computation time as a function of the number of processors. We note that, our experiments are performed under laboratory conditions and is meant for illustrative purposes. In practice, the command center is endowed with far more powerful computational tools, which are indeed necessary for realtime operations. Using such tools, the performance of our scheduling scheme can be significantly improved, and can be fine-tuned to satisfy the stringent requirements on latency for transmission applications on the grid (see, for example, \cite{Kansal2012}). In order to avoid obscuring the main topic of the paper, we have not investigated computational implementation in greater detail. This work is a first step in the direction of devising scheduling algorithms for phasor data transmission, with scope for improvement in both theory and practice.

\vspace{-0.1in}
\section{Conclusion and future work}\label{sec:conclusion}
We have devised a scheduling algorithm for ordering transmission of phasor data from the substations to the control center in as short a time frame as possible on the hierarchical information architecture in the electric power grid. The problem was cast in the framework of classic single machine job scheduling with precedence constraints. The optimal number of PMUs was obtained by solving the PMU placement problem with the objective of complete network observability ignoring zero injection measurements. The weights associated with the PMUs and the precedence constraints between the PMUs were assumed known. The processing time of the PMUs were chosen uniformly at random from a pre-specified set of times. The problem is provably NP-hard, so approximation algorithms are typically used to obtain polynomial time complexity albeit being suboptimal. Using branching and bounding techniques, we arrive at a lower bound on the feasible schedule. The average computation time and the average number of nodes in the search tree was plotted as a function of the number of PMUs installed on the grid, and were compared with the performance of a simple greedy algorithm. As expected, for fewer PMUs on the grid, the performance of scheduling is favorable, however, for more than 200 PMUs, the computation time increases at an alarming rate. Future work would involve (i) devising more efficient cycle elimination techniques to reduce the gap between the lower and upper bounds and (ii) problem formulation for distributed scheduling to facilitate distributed functionality of the existing control centers \cite[Section I D]{Bose2010}.

\vspace{-0.1in}
\section*{Acknowledgement}
K. G. Nagananda would like to thank Chandra R. Murthy, at IISc, for providing the lab space during the course of this work. The work of Pramod Khargonekar was supported in part by the NSF under Grant CNS 12-39274. Sincere thanks to the anonymous referees for their invaluable suggestions.

\vspace{-0.05in}
\bibliographystyle{IEEEtran}
\bibliography{IEEEabrv,proposals}
\raggedbottom

\end{document}